\documentclass[lettersize,journal]{IEEEtran}
\usepackage{setspace}
\usepackage{lipsum}
\usepackage{cite}
\usepackage{pgfplots}

\pgfplotsset{compat=1.10}

\usepgfplotslibrary{fillbetween}

\usepackage{bm}
\usepackage{lipsum}
\usepackage{makeidx}
\usepackage{enumerate}
\usepackage{color}
\usepackage{cite}
\usepackage{amsmath,amsthm}   
\usepackage{amssymb}
\usepackage{nomencl}
\usepackage{multirow}
\usepackage{graphicx}
\usepackage{epstopdf}
\usepackage{threeparttable}
\usepackage{multicol}
\usepackage{import}
\DeclareGraphicsExtensions{.pdf,.jpeg,.png,.jpg,.emf,.eps}
\usepackage{calc}
\usepackage{here}
\usepackage{url}

\usepackage{upref}
\usepackage{comment}
\usepackage{times}
\usepackage{dsfont}
\usepackage{epic,eepic}
\usepackage{rawfonts}
\usepackage[T1]{fontenc}
\usepackage{latexsym}
\usepackage{amsfonts}

\usepackage{capitalgreekitalic}
\usepackage[mode=buildnew]{standalone}

\usepackage[font=small,labelfont=bf]{caption}
\usepackage[font=small,labelfont=bf]{subcaption}

\usepackage{xcolor}
\usepackage{tikz,pgfplots}
\usetikzlibrary{calc}
\usetikzlibrary{intersections}
\usetikzlibrary{arrows,shapes}
\usetikzlibrary{positioning}
\usepackage[ruled,noline,linesnumbered]{algorithm2e}

\newtheorem{lemma}{Lemma}

\newtheorem{remark}{Remark}

\theoremstyle{definition}

\def\R{{\bf R}}
\def\Q{{\bf Q}}
\def\I{{\bf I}}
\def\A{{\bf A}}
\def\U{{\bf U}}
\def\V{{\bf V}}
\def\F{{\bf F}}
\def\T{{\bf T}}
\def\G{{\bf G}}
\def\H{{\bf H}}
\def\U{{\bf U}}

\def\C{{\bf C}}
\def\r{{\bf r}}
\def\s{{\bf s}}

\def\x{{\bf x}}
\def\y{{\bf y}}
\def\z{{\bf z}}
\def\n{{\bf n}}
\def\0{{\bf 0}}
\def\Thetab{\bm{\Theta}}
\def\Sigmab{\bm{\Sigma}}
\def\Lambdab{\bm{\Lambda}}

\def\tr{\operatorname{tr}}
\def\det{\operatorname{det}}
\def\diag{\operatorname{diag}}

\def\vec{\operatorname{vec}}
\def\unvec{\operatorname{unvec}}
\def\Exp{\operatorname{E}}
\def\blkdiag{\operatorname{blkdiag}}

\begin{document}
\title{{Interference Minimization in Beyond-Diagonal RIS-assisted MIMO Interference 
Channels}
\author{Ignacio~Santamaria, \IEEEmembership{Senior Member, IEEE}, Mohammad Soleymani, \IEEEmembership{Member, IEEE}, Eduard Jorswieck, \IEEEmembership{Fellow, IEEE}, Jes{\'u}s Guti{\'e}rrez \IEEEmembership{Member, IEEE}}
\thanks{I.Santamaria is with the Department
of Communications Engineering, Universidad de Cantabria, 39005 Santander, Spain (e-mail: i.santamaria@unican.es).}
\thanks{M. Soleymani is with the Signal and System Theory Group, Universit{\"a}t  Paderborn, 33098 Paderborn, Germany (e-mail: mohammad.soleymani@uni-paderborn.de).}
\thanks{Eduard Jorswieck is with Institute for Communications Technology, Technische Universit{\"a}t Braunschweig, 38106 Braunschweig,
Germany (email: e.jorswieck@tu-bs.de).}
\thanks{J. Guti{\'e}rrez is with IHP - Leibniz-Institut
f{\"u}r Innovative Mikroelektronik, 15236 Frankfurt (Oder), Germany (email: teran@ihp-microelectronics.com).}
}
\maketitle

\begin{abstract}
This paper proposes a two-stage approach for passive and active beamforming in multiple-input multiple-output (MIMO) interference channels (ICs) assisted by a beyond-diagonal reconfigurable intelligent surface (BD-RIS). In the first stage, the passive BD-RIS is designed to minimize the aggregate interference power at all receivers, a cost function called interference leakage (IL). To this end, we propose an optimization algorithm in the manifold of unitary matrices and a suboptimal but computationally efficient solution. In the second stage, users' active precoders are designed under different criteria such as minimizing the IL (min-IL), maximizing the signal-to-interference-plus-noise ratio (max-SINR), or maximizing the sum rate (max-SR). The residual interference not cancelled by the BD-RIS is treated as noise by the precoders. Our simulation results show that the max-SR precoders provide more than $20 \%$ sum rate improvement compared to other designs, especially when the BD-RIS has a moderate number of elements ($M<20$) and users transmit with high power, in which case the residual interference is still significant.

\end{abstract}

\begin{IEEEkeywords}
 Beyond-diagonal reconfigurable intelligent surface, multi-antenna communications, interference leakage, interference channel, manifold optimization.
\end{IEEEkeywords}

\section{Introduction}

\subsection{Background}

The potential of reconfigurable intelligent surfaces (RISs) as a PHY-layer 6G enabling technology that can improve coverage, transmission rates, energy efficiency, or network latency is being recognized and assessed by academia and industry alike. A general overview of this research topic and the main applications of RISs in wireless communication systems can be found at \cite{RenzoJSAC2020}. Relevant use cases, deployment scenarios, and operational requirements for each identified use case are reported in \cite{ETSIreport}. Initial studies and implementations of RISs have been focused on passive lossless single-connected structures, leading to diagonal scattering matrices whose elements are modeled as phase shifters. The joint optimization of the precoders and the diagonal RIS in a multi-cell MIMO system to maximize the weighted sum rate is proposed in \cite{pan2020multicell}. Based on a realistic model for the RIS power consumption, the problem of maximizing the energy efficiency is formulated in \cite{ZapponeTWT2019}, and an algorithm for optimizing the RIS elements and the power transmitted to the users in a downlink channel is proposed. The authors in \cite{ZhangTWT2019} consider the optimization of the active precoders and the passive reflective RIS to minimize the power transmitted by an access point or base station subject to individual quality of service (QoS) constraints for users.

More recently, the concept of beyond-diagonal RIS (BD-RIS) has been proposed as a generalization in which all RIS elements can be connected by variable reactances to each other, leading to fully-connected or group-connected (block-diagonal) scattering matrices \cite{ClerckxTWC22b}, with different trade-offs between complexity and performance. Analyzing the BD-RIS as a multi-port network, in \cite{ClerckxTWC22a} it is shown that the BD-RIS scattering matrix must be unitary and symmetric, which introduces new challenges in its optimization.

\subsection{Related work on BD-RIS optimization}

Over the last few years, different algorithms have been proposed to optimize the BD-RIS elements to maximize, among other cost functions, the received signal-to-noise ratio (SNR), the rate, or the energy efficiency. Closed-form fully-connected BD-RIS solutions that maximize the equivalent channel gain exist for single-input single-output (SISO) and MISO/SIMO channels \cite{SantamariaSPLetters2023, NeriniTWC2023, MaoCL2024}. In \cite{ClerckxTWC24a},\cite{ClerckxJSAC23}, the authors consider a multi-user multiple-input single-output (MU-MISO) downlink channel assisted by a multi-sector BD-RIS in which the RIS matrices of each cell are single-connected and therefore diagonal. The optimization of multi-sector BD-RIS coefficients is performed on the complex sphere manifold. Li {\it et al} proposed in \cite{ClerckxTWC22b} a manifold optimization algorithm to maximize the sum rate in BD-RIS-assisted MU-MISO systems when the direct channels are blocked. A similar algorithm is described in \cite{MishraCL2024}. The successive convex approximation technique has been used in \cite{SoleymaniTWC2023} for optimizing the spectral efficiency and energy efficiency of a multiple access system assisted by a BD-RIS with a group-connected architecture of group size two. An optimization framework for BD-RIS based on the penalty dual decomposition (PDD) method \cite{ShiTSP20PDD} for the optimization of different cost functions is proposed in \cite{ClercksTWC24Newton}. Although this algorithm is rather general, it only considers single antenna receivers (MU-MISO) and has a high computational cost. Moreover, the symmetry constraint of the BD-RIS matrix is not explicitly considered in \cite{ClercksTWC24Newton}. Geodesic manifold optimization algorithms have recently been proposed in \cite{ZhaoArxiv24}, \cite{SantamariaSPAWC24} to maximize the rate in BD-RIS-aided point-to-point MIMO channels. While \cite{ZhaoArxiv24} only considers the unitary constraint for the BD-RIS matrix, \cite{SantamariaSPAWC24} also considers the symmetry constraint. In \cite{YahyaArxiv24} a closed-form BD-RIS solution is found by projecting the Maximum Ratio Transmission (MRT) precoder (to maximize the signal power) or the zero-forcing precoder (to cancel interference) into the set of symmetric and unitary matrices. This relax-then-project approach, also used in \cite{MaoCL2024}, may cause a significant performance degradation compared to solving the original constrained problem. The method presented in this paper avoids the projection on the set of symmetric matrices, thus improving the results.

\subsection{Motivation}

Most of the aforementioned BD-RIS-assisted scenarios mainly consider point-to-point or MU-MISO systems. However, works on the $K$-user MIMO interference channel (IC) are more scarce and have been mostly limited to diagonal RIS. To achieve the maximum degrees of freedom (DoF) of the $K$-user IC, the interfering signals at each receiver must fall into a reduced-dimensional subspace, a technique called {\it interference alignment} (IA) \cite{Cadambe2008IA}. The design of IA multi-antenna precoders uses as optimization criterion the sum of the power of the interference leaked in the receive subspaces of all users, a cost function that is known in the literature as {\it interference leakage} (IL) \cite{gomadam2011distributed, Gonzalez14}. Therefore, most studies on RIS-assisted ICs aim at canceling or aligning the interference, thus maximizing the channel DoFs. 

The first work in this line is \cite{jiang2022interference}, where the ability of an RIS to completely null the interference simultaneously on all the receivers in an IC is studied for the first time. However, the scenario discussed in \cite{jiang2022interference} is rather limited, considering only a SISO-IC with blocked direct channels. In \cite{FuICC2021}, the authors optimize an RIS to maximize the DoF of the $K$-user MIMO-IC. In \cite{AbradoCL2021}, the $K$-user MIMO-IC sum rate, including mutual coupling effects, is maximized using a weighted minimum mean squared error (wMMSE) approach. In \cite{SchoberArxiv2021} the authors study the DoF of the $K$-user RIS-assisted SISO IC with symbol extensions (i.e., for time-varying channels). In \cite{SantamariaICASSP23} an algorithm for interference minimization in the $K$-user MIMO-IC assisted by a passive lossless RIS was proposed. An active (unconstrained) RIS to perform IA in the $K$-user MIMO-IC has been considered in \cite{LiuCL2024}. The active diagonal RIS reduces the size of the interference subspace at each receiver while zero-forcing precoders are used to null the multiuser interference. In this work, we consider more realistic passive architectures.

Certain important reasons justify minimizing the IL in a BD-RIS-assisted $K$-user MIMO-IC. First, for the $K$-user MIMO-IC the min-IL precoders attain (if the IA problem is feasible) the maximum DoF of the channel. Second, if a BD-RIS can sufficiently suppress interference, the users could design their precoders independently treating the residual interference as noise (TIN) \cite{geng2015optimality}. Finally, the IL yields quadratic functions that are usually easier to optimize than other cost functions commonly used, such as the sum rate or the mean-squared error (MSE). Motivated by these considerations, in this paper we consider the IL as a suitable metric for the BD-RIS design in MIMO-ICs.
However, the literature has not studied algorithms that optimize the IL using BD-RIS architectures subject to different power constraints and using various types of precoders. This is the problem that motivates this work.

\subsection{Contributions}
This paper proposes a two-stage approach for interference management in the $K$-user MIMO-IC assisted by a BD-RIS. The IL is minimized in the first stage by the BD-RIS, while the precoder design takes care in the second stage of the residual interference that the BD-RIS could not remove. The contributions of this work can be summarized as follows:

\begin{enumerate}
    \item Unlike most of the works on BD-RIS that considers MU-MISO systems \cite{MishraCL2024, ClercksTWC24Newton, ClerckxTWC22b, SoleymaniTWC2023}, this paper focuses on the $K$-user MIMO-IC, which requires specific precoding and interference management strategies such as IA to achieve the maximum channel's DoF.

    \item For this scenario, we consider for the first time a two-stage approach for designing the BD-RIS and the precoders. In the first stage, the BD-RIS minimizes the IL to reduce the interference level below the noise level thus facilitating the design of the precoders in the second stage. The precoders are designed in the second stage to optimize different criteria such as maximizing the sum rate, the SINR, or minimizing the IL.

    \item For a fully-connected BD-RIS architecture we propose a new optimization algorithm in the manifold of unitary matrices. Unlike existing alternatives in the literature, the proposed manifold optimization (MO) algorithm incorporates the symmetry constraint in the solution by taking advantage of the reparametrization of any symmetric and unitary matrix through Takagi's factorization.

    \item We propose a relax-then-project (RtP) algorithm that provides a suboptimal solution with lower computational complexity than the MO algorithm. The RtP method first solves a relaxed version of the problem by replacing the unitary constraint on the BD-RIS matrix with a constraint on its trace. The relaxed problem incorporates the symmetry constraint of the BD-RIS into the problem through the commutation matrix. We believe this contribution might be useful in other problems involving symmetric BD-RISs. After solving the resulting convex problem, the relaxed solution is projected on the set of unitary matrices using the method proposed in \cite{MaoCL2024,ZhouCL2025}.

    \item To reduce the computational and circuit complexity associated with the fully connected BD-RIS architecture, we adapt the IL minimization algorithm to a group-connected architecture \cite{ClerckxTWC22a, ClerckxTWC22b}. The group-connected algorithm sequentially optimizes each group using either the MO algorithm or the RtP method.

    \item A theoretical analysis for unconstrained BD-RIS (probably using active components), allows us to estimate the minimum number of elements needed to cancel the interference simultaneously at all receivers, thus generalizing previous results for diagonal RIS.

    \item Simulation results compare the performance of the proposed two-stage approach in terms of interference-to-noise ratio (INR) levels and achievable rates for different BD-RIS locations and precoder designs. A fully-connected BD-RIS with a moderately high number of elements and optimally positioned significantly reduces interference levels. The group-connected BD-RIS optimized with the MO algorithm and with group size $M_g=8$ provides a reasonable tradeoff between performance and complexity. The precoders take care of the residual interference through proper design. The max-SR precoders provide more than 20 $\%$ sum rate improvement compared to other designs especially when the BD-RIS has a moderate number of elements ($M<20$) and users transmit with high power, in which case the residual interference is still significant.

\end{enumerate}

\subsection{Paper outline and notations}
The rest of the paper is organized as follows. Section \ref{sec:ILcost} presents the IL minimization problem in the $K$-user MIMO-IC and briefly introduces the two-stage optimization framework. Section \ref{sec:StageI} describes Stage I, where the BD-RIS is optimized to minimize the IL. We present a manifold optimization algorithm based on Takagi decomposition \cite{Takagi} for the fully-connected architecture, and a relax-then-project suboptimal solution. Optimization algorithms for the group-connected BD-RIS and the conventional diagonal RIS are also described in  Sec. \ref{sec:StageI}. Section \ref{sec:StageII} describes Stage II of the method, where the users' precoders are optimized according to different criteria, such as minimizing the IL, maximizing the signal-to-interference-plus-noise (SINR), or maximizing the sum rate. An alternating optimization (AO) algorithm to jointly design precoders and BD-RIS minimizing the IL is also described in  Section \ref{sec:StageII}. Section \ref{sec:simulations} compares the performance of the different BD-RIS and RIS designs in terms of the INR and the sum rate achieved in several scenarios. This work's main conclusions and future directions are described in Sec. \ref{sec:conclusions}.

\textit{Notation}: In this paper, matrices are denoted by bold-faced upper case letters, column vectors are denoted by bold-faced lower case letters, and scalars are denoted by light-faced lower case letters. The superscripts $(\cdot)^T$, $(\cdot)^*$, and $(\cdot)^H$ denote transpose, conjugate, and Hermitian conjugate, respectively. The trace and determinant of a matrix $\A$ will be denoted, respectively, as $\tr(\A)$ and $\det(\A)$. $\I_n$ denotes the identity matrix of size $n$, ${\cal CN}(0,1)$ denotes a complex proper Gaussian distribution with zero mean and unit variance, and $\x \sim {\cal CN}_{n}({\bf 0}, {\bf R})$ denotes a complex Gaussian vector in $\mathbb{C}^n$ with zero mean and covariance matrix ${\bf R}$. $\Exp[\cdot]$ denotes mathematical expectation. In addition, $\odot$ denotes Hadamard product and $\otimes$ denotes Kronecker product. The $\diag$ operator of a matrix, $\diag(\A)$, extracts the elements of the main diagonal of the square matrix $\A$ and stores them in a column vector, and $\vec(\A)$ denotes the vectorization operator which converts the $M\times M$ matrix $\A$ into an $M^2 \times 1$ column vector by stacking the columns of $\A$ on top of one another. $\nabla_{\Q} J$ denotes the gradient of the matrix-valued function $J$ evaluated at $\Q$. Finally, $\mathcal{O}(\cdot)$ is the big-O notation for representing the computational complexity of algorithms.

\section{BD-RIS-assisted MIMO IC}
\label{sec:ILcost}
\subsection{Signal model}
As illustrated schematically in Fig. \ref{fig:MIMOIC}, we consider a $K$-user MIMO interference channel (MIMO-IC) assisted by a BD-RIS designed to suppress the aggregate interference power at all receivers. The $k$th link has $N_{T_k}$ transmit antennas, $N_{R_k}$ receive antennas, and transmits $d_k$ streams. The BD-RIS has $M$ elements and all Txs and Rxs are in the reflection space of the BD-RIS. According to the commonly used notation, we denote the MIMO-IC in abbreviated form as $(N_{T_k} \times N_{R_k}, d_k)^K$ \cite{gomadam2011distributed,gonzalez2014feasibility}. The equivalent MIMO channel from the $l$th transmitter to the $k$th receiver is
\begin{equation}
\widetilde{\H}_{lk}= \H_{lk} + \F_k \Thetab \G_l^H,
\label{eq:Eqchannel}
\end{equation}
where $\H_{lk} \in \mathbb{C}^{N_{R_k}\times N_{T_l}}$ is the $(l,k)$ MIMO-IC, $\G_l \in \mathbb{C}^{N_{T_l} \times M }$ is the channel from the $l$th transmitter to the BD-RIS, $\F_k \in \mathbb{C}^{N_{R_k} \times M }$ is the channel from the BD-RIS to the $k$th receiver, and $\Thetab$ is the $M\times M$ BD-RIS scattering matrix. The $k$th user transmits a precoded signal $\x_k = \V_k \s_k$, where $\V_k \in \mathbb{C}^{N_{T_k}\times d_k}$ is a precoder scaled to satisfy the total power constraint and $\s_k$ contains the information symbols, $\s_k \sim {\cal{CN}} \left(\0, \I_{d_k} \right)$. 
The precoders determine the beamforming directions or beams and the allocated powers to the $d_k$ streams. 
Therefore, the total power transmitted by user $k$ is $P_{t_k} = \tr(\Exp[\x_k \x_k^K]) = \tr(\V_k \V_k^H)$. The signal received by the $k$th user is
\begin{equation}
    \y_{k} = \widetilde{\H}_{kk} \x_k + \sum_{l=1,l\neq k}^K \widetilde{\H}_{lk} \x_l + \n_k,
\end{equation}
where the second term is the interference caused by other users and $\n_k \sim {\cal{CN}} (\0, \sigma^2\I_{N_{R_k}})$ is the additive white Gaussian noise (AWGN). The $k$th user applies a unitary decoder $\U_k \in \mathbb{C}^{N_{R_k}\times d_k}$ to obtain 

\begin{figure}[t!]
\centering		
\includegraphics[width=.5\textwidth]{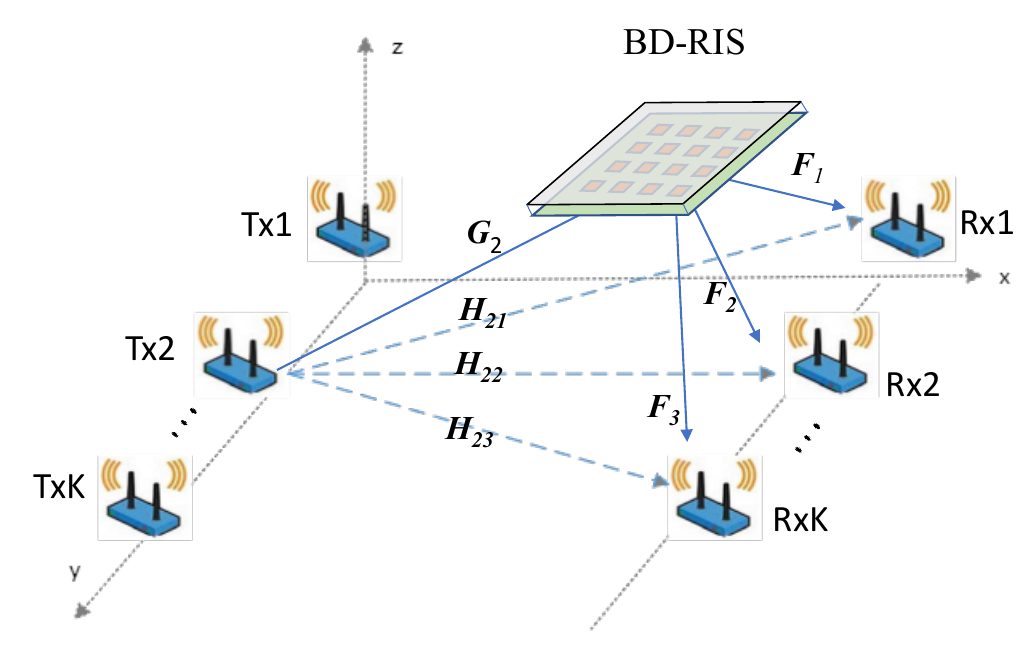}
\caption{BD-RIS-assisted $K$-user MIMO IC. Dashed lines represent the direct Tx-Rx links; solid lines represent the channels through the BD-RIS.} 
\label{fig:MIMOIC}	
\end{figure}
\begin{equation*}
    \z_k = \underbrace{\U_k^H \widetilde{\H}_{kk} \V_k \s_k}_{{\rm Signal}} + \underbrace{\U_k^H \left(\sum_{l=1,\neq k}^K \widetilde{\H}_{lk} \V_l \s_l \right)}_{{\rm Interference}} + \underbrace{\U_k^H \n_k}_{{\rm Noise}}. 
\end{equation*}
As the noise is Gaussian and isotropically distributed, $\n_k \sim {\cal{CN}} (\0, \sigma^2\I_{N_{R_k}})$, and the decoder is unitary, $\U_k^H \U_k = \I_{d_k}$, then
$\U_k^H\n_k \sim {\cal{CN}} (\0, \sigma^2\I_{d_k})$ \cite[App. D]{Coherence}. In certain cases, we will find it useful to absorb the precoders and decoders into a new equivalent channel from the $l$th transmitter to the $k$th receiver as
\begin{eqnarray}
\U_k^H \widetilde{\H}_{lk} \V_l & = & \U_k^H\H_{lk}\V_l + \U_k^H\F_k \Thetab \G_l^H\V_l \nonumber \\
& = & \overline{\H}_{lk} + \overline{\F}_k \Thetab \overline{\G}_l^H,
\label{eq:EqchannelwithUV}
\end{eqnarray}
where $\overline{\F}_k = \U_k^H \F_k $ is the $d_k \times M $ equivalent channel from the BD-RIS to the receiver after decoding, $\overline{\G}_l =  \V_l^H \G_l$ is the $d_l \times M$ equivalent channel from the transmitter to the BD-RIS after precoding, and $\overline{\H}_{lk}=\U_k^H \H_{lk} \V_l$ is the $d_k \times d_l$ equivalent channel matrix after precoding-decoding. Table \ref{tab:symbols} shows the most frequently used symbols and their definitions. 
\begin{table}
\centering
\small
\begin{tabular}{ ||c|c||} 
\hline
\rule{0pt}{3ex} 
   $\H_{kl}$  & $(l,k)$ MIMO-IC direct link \\ [0.5ex] 
\hline
\rule{0pt}{3ex} 
  $\widetilde{\H}_{kl}$ & $(l,k)$ Equivalent BD-RIS-assisted MIMO-IC \\ 
\hline  
\rule{0pt}{3ex} 
 $\overline{\H}_{kl}$ & $(l,k)$ MIMO-IC after precoding+decoding\\ 
 \hline  
 \rule{0pt}{3ex} 
  $\F_{k}$ & Channel from BD-RIS to $k$th user Rx\\ 
\hline 
\rule{0pt}{3ex}
   $\overline{\F}_{k}$ & Channel from BD-RIS to $k$th user Rx after decoding \\ 
\hline
\rule{0pt}{3ex} 
$\G_{l}$ & Channel from $l$th user Tx to BD-RIS \\ 
\hline 
\rule{0pt}{3ex}
   $\overline{\G}_{l}$ & Channel from $l$th user Tx to BD-RIS after precoding \\  
 \hline
 \rule{0pt}{3ex}
   $\V_k,\U_k$ & $k$th user precoder and decoder  \\  
 \hline
\end{tabular}
\caption{Symbols used frequently and their definitions.}
\label{tab:symbols}
\end{table}

\subsection{Proposed two-stage optimization framework}
Any cost function optimized on the $K$-user MIMO IC requires the joint optimization of the precoders/decoders and the BD-RIS scattering matrix, a procedure usually carried out through AO algorithms with a high computational cost. We adopt a different approach that decouples the problem in two stages. The passive BD-RIS is designed in the first stage to reduce the total aggregate interference power at the receivers (a cost function in which the direct channels do not intervene). For the second stage, the (active) precoders and decoders are designed in the second stage, optimizing another cost function that improves the quality of the equivalent direct channels for each user. Two-stage optimization procedures have been recently proposed in BD-RIS assisted MU-MISO communication networks in \cite{MaoCL2024,YahyaArxiv24}. However,  this work is the first to propose a two-stage approach for the $K$-user MIMO-IC.

The proposed two-stage framework entails a reduction in complexity, owing to the decoupling of passive and active beamforming problems. Another motivation is the architectural evolution from 5G to 6G systems, whereby the BD-RIS infrastructure provider might be different from the radio access network (RAN) operator in charge of the users' resource allocation strategies. 

\section{Stage I: BD-RIS design to minimize interference}
\label{sec:StageI}
\subsection{Problem statement}
The BD-RIS is a passive multiport network, so the scattering matrix $\bm{\Theta}$ is modeled as a symmetric and unitary $M\times M$ matrix, $\Thetab^T = \Thetab$ and $\Thetab^H \Thetab = \I_M$. This architecture, in which tunable impedances connect all pairs of ports, is known as fully-connected BD-RIS \cite{ClerckxTWC22a}. Despite having a high implementation complexity, the fully-connected BD-RIS architecture provides the greatest flexibility in optimizing the various cost functions.

In the first stage, the BD-RIS minimizes the IL \cite{gomadam2011distributed}, defined as the sum of the squared Frobenius norms of all ICs. The BD-RIS optimization problem is\footnote{To simplify the notation of the remainder of the paper we will denote the double summation $\sum_{k=1}^K\sum_{l = 1, l \neq k}^K$ simply as $\sum_{l \neq k}$.}

\begin{subequations}
\begin{align}
({\cal P}_1): \,\min_{\Thetab}\,\,& \sum_{k=1}^K\sum_{l = 1, l \neq k}^K \| \H_{lk} + \F_k \Thetab \G_l^H \|_F^2 \label{eq:minIL}  \\
\text{s.t.}\,\,& \Thetab^T = \Thetab, \, \Thetab^H \Thetab = \I_M.
\end{align}
\end{subequations}

The IL cost function \eqref{eq:minIL}, minimized in Stage I by the BD-RIS, is independent of the precoders and decoders to be designed in Stage II and can be interpreted as a worst-case for the interference power observed with a given set of precoders and decoders.

\subsection{Unconstrained BD-RIS}
Although this paper focuses on passive BD-RIS, it may be illustrative to analyze the problem with an unconstrained, probably active, BD-RIS. When the number of elements is sufficiently large, an unconstrained BD-RIS can make the IL cost function in \eqref{eq:minIL} zero, thus perfectly neutralizing the interference at all receivers. The following lemma presents this result.

\begin{lemma}\label{Prop:LemmaUnc} 
Consider a generic\footnote{Meaning that all channels are independent of each
other and their entries are also independently drawn from a
continuous distribution.} $K$-user MIMO-IC $(N_{T_k} \times N_{R_k}, d_k)^K$ assisted by an unconstrained BD-RIS, $\Thetab$, optimized to minimize the IL cost function in \eqref{eq:minIL}. The only assumption is that $M \geq \max(N_{T_k}, N_{R_k}) \, \forall k$. An unconstrained BD-RIS can make the IL zero if the number of elements satisfies $M^2 \geq \sum_{l \neq k} N_{R_k} N_{T_l}$.
\end{lemma}

\begin{proof}

It is a simple algebra exercise to show that the IL cost function can be rewritten as the following quadratic function 
\begin{equation}
IL(\r) = \tr(\T) + \r^H \Sigmab \r + 2 {\rm Re}(\r^H {\bf s}),
\label{eq:ILcostfunction1}
\end{equation}
where $\r = \vec(\Thetab)$, $\s = \vec \left( \sum_{l \neq k} \F_k^H \H_{lk}\G_l \right)$, $\Sigmab = \sum_{l \neq k}  \G_l^T \G_l^* \otimes 
 \F_k^H\F_k $, and  $\T =\sum_{l \neq k} \H_{lk}^H \H_{lk}$.
Notice first that ${\bf s} \in {\rm colspan}(\Sigmab)$ and $\Sigmab$ is an $M^2\times M^2$ positive semidefinite (psd) matrix. The matrices $ \F_k^H\F_k$ and $\G_l^T \G_l^*$ are, respectively, rank-$N_{R_k}$ and rank-$N_{T_l}$ psd matrices. Then, $\Sigmab$ has rank $ \min(M^2,\sum_{k \neq l} N_{R_k} N_{T_l} )$ \cite{Horn2020}. This means that without a power constraint on the BD-RIS elements, the IL can always be made zero regardless of the precoders or decoders with a BD-RIS with $M^2 >\sum_{k \neq l} N_{R_k} N_{T_l}$ active elements. If the BD-RIS is symmetric but otherwise unconstrained, we only have to replace $M^2$ by $M(M+1)/2$ in the inequality, as $M(M+1)/2$ is the dimension of the space of $M\times M$ symmetric matrices.
\end{proof}

\begin{remark}
When all users have the same number of antennas and the links are $N_R \times N_T$ MIMO channels, the condition in Lemma \ref{Prop:LemmaUnc} reduces to $M^2 > K(K-1) N_RN_T$, thus generalizing the zero IL condition for unrestricted (active) diagonal RIS-assisted MIMO ICs derived in \cite{SantamariaICASSP23} and \cite{LiuCL2024}, which is $M > K(K-1) N_RN_T$. The condition for the $K$ user SISO IC, $M^2 > K(K-1)$, was derived in \cite{jiang2022interference}. In addition, when the precoders and decoders have been designed in a previous step and the unconstrained BD-RIS is designed afterward, the number of elements needed for zero IL relaxes to $M^2 >\sum_{k \neq l} d_{k} d_{l}$. When all users transmit $d$ streams, the condition becomes $M^2 > K(K-1)d^2$.   
\end{remark}

\subsection{Manifold optimization algorithm}
\label{sec:plBDRISalg}
The IL minimization problem ${\cal {P}}_1$ is non-convex and requires an iterative algorithm to find a (probably suboptimal) solution. In this subsection, we propose an iterative algorithm on the manifold of the unitary matrices similar to that in \cite{SantamariaSPAWC24} to maximize the capacity in a point-to-point MIMO link. 

The main idea of the algorithm proposed in \cite{SantamariaSPAWC24} is to apply Tagaki's factorization \cite{Takagi, Autonne, SantamariaSPLetters2023}, which proves that any complex {\it symmetric} matrix $\A$ can be factored as $\A = \Q \bm{ \Sigma} \Q^T$, where $\Q$ belongs to the manifold of complex unitary matrices denoted here as ${\cal{U}}(M)$, and $\bm{ \Sigma}$ is a diagonal matrix containing the singular values of $\A$ along its diagonal. The BD-RIS scattering matrix, $\Thetab$, is symmetric and unitary. Therefore, all its singular values are one and it can be factored as $\Thetab = \Q \Q^T$. We use this fact to formulate the IL minimization problem as follows\footnote{The same formulation would be valid for a given set of precoders and decoders by replacing the matrices $(\H_{lk},\F_k,\G_l)$ by $(\overline{\H}_{lk}, \overline{\F}_k, \overline{\G}_l)$.}
\begin{align}
\label{eq:ILBDRISpl}
({\cal P}_{1a}): \min_{\Q \in {\cal{U}}(M) }\,&  \sum_{l \neq k} \| \H_{lk} + \F_k \Q \Q^T \G_l^H \|_F^2.
\end{align}
To solve $({\cal P}_{1a})$, we perform the optimization on the unitary group \cite{boumal2023intromanifolds}. The tangent plane at a point $\Q \in {{\cal{U}}(M)}$, denoted as $\Delta \Q$, comprises all $M\times M$ matrices such that $\Q^H \Delta \Q$ is skew-Hermitian. The complex matrix derivative of \eqref{eq:ILBDRISpl} with respect to $\Q^*$ is
\begin{equation}
    \nabla_{\Q^*} J = \sum_{l \neq k}  \F_k^H \left( \F_k \Q\Q^T \G_l^H + \H_{lk} \right) \G_l \Q^*.
    \label{eq:unc_gradient}
\end{equation}
The projection of the unconstrained gradient onto the tangent space at $\Q$ is $\pi_T (\nabla_{\Q^*}  J) = \Q {\bf B}_{skew}$, where
\begin{equation}
{\bf B}_{skew} = \left((\nabla_{\Q^*} J)^H \Q - \Q^H (\nabla_{\Q^*} J)  \right)/2
\label{eq:Sskew}
\end{equation}
is skew-Hermitian. Finally, to update the unitary matrix we move along the geodesic starting at $\Q$ (the value at the current iteration) with direction $\nabla_{\Q^*} J = \Q {\bf B}_{skew}$ as 
\begin{equation}
    \Q = \Q e^{\mu {\bf B}_{skew}},
    \label{eq:geodesic}
\end{equation}
where $e^{\A}$ denotes here the matrix exponential. The learning step size, $\mu >0$, can be conveniently chosen and adapted using a line search procedure. A summary of the min-IL MO algorithm is shown in Algorithm \ref{alg:BDRISopt}.

\begin{algorithm}[!t]
\small
\DontPrintSemicolon
\SetAlgoVlined
\KwIn{Initial $\Q \in {\cal{U} (M)}$; $\H_{lk}, \F_k, \G_l$, $\forall k \neq l$, learning rate $\mu$}
\KwOut{Final BD-RIS $\Thetab = \Q\Q^T$}
\While{MO convergence criterion not true}{
Calculate $\nabla_{\Q^*} J$ as \eqref{eq:unc_gradient}  \;
Find ${\bf B}_{skew} $ as  \eqref{eq:Sskew}\;
Update $\Q$ as \eqref{eq:geodesic}
}
\caption{{\small Min-IL MO BD-RIS}}
\label{alg:BDRISopt}
\end{algorithm}

\paragraph{Computational complexity}
The complexity per iteration of the MO algorithm for a $(N_{T} \times N_{R}, d)^K$ MIMO-IC is as follows. The computation of the $M \times M$ unconstrained gradient in \eqref{eq:unc_gradient} has complexity $\mathcal{O}(K(K-1) N_R^2 N_T^2 M^3)$. The computation of the skew-Hermitian matrix in \eqref{eq:Sskew} and the matrix exponentiation in \eqref{eq:geodesic} both have $\mathcal{O}(M^3)$. Hence, the overall complexity is $\mathcal{O}((K(K-1) N_R^2 N_T^2 + 2) M^3) = \mathcal{O}(M^3)$ per iteration.

\subsection{A relax-then-project suboptimal solution}
\label{sec:relaxed}
The computational complexity of the proposed MO algorithm, which is cubic on the number of BD-RIS elements per iteration, may be excessive for large $M$. Therefore, it is of interest to explore other suboptimal solutions that may be computationally more efficient. In this subsection, we solve a relaxed version of the problem \eqref{eq:minIL} that replaces the unitary constraint $\Thetab^H \Thetab = \I_M$  with a convex constraint on the trace $\tr(\Thetab) \leq M$. The solution of the relaxed problem on the set of complex and symmetric matrices admits a closed solution as a regularized least squares problem. The relaxed solution is then projected onto the set of unitary matrices to obtain the final suboptimal solution using the approach proposed in \cite{MaoCL2024,ZhouCL2025}. We denote this algorithm as relax-then-project (RtP).
The relaxed convex optimization problem is
\begin{subequations}
\begin{align}
 \,\min_{\Thetab} \,\,& \sum_{l \neq k} \|  \H_{lk} + \F_{k} \Thetab \G_{l}^H \|_F^2 \label{eq:minILbisa}  \\
\text{s.t.} \,\, & \Thetab^T = \Thetab, \, \tr(\Thetab^H \Thetab) \leq M \label{eq:minILbisb}.           
\end{align}
\end{subequations}

\noindent Now we describe how to incorporate the symmetry constraint into the problem and transform it into a standard Quadratically Constrained Quadratic Program (QCQP) problem. As described in the proof of Lemma \ref{Prop:LemmaUnc}, the MinIL cost function can be expressed as a quadratic function of the vectorized BD-RIS coefficients $\r = \vec(\Thetab)$
\begin{equation*}
    \sum_{l\neq k} \|  \H_{lk} + \F_{k} \Thetab \G_{l}^H \|_F^2 = 
\tr(\T) + \r^H \Sigmab \r + 2 {\rm Re}(\r^H {\bf s}),
\end{equation*}
where $\T$, $\s$, and $\Sigmab$ are defined as in \eqref{eq:ILcostfunction1}.
 
The linear symmetry constraint $\Thetab^T = \Thetab$ in \eqref{eq:minILbisb} can be written in terms of $\r$ through the $M^2 \times M^2$ commutation matrix\footnote{A Matlab snippet to generate ${\bf P}$ is ${\rm A = reshape(1:M*M, M, M)};$ 
${\rm v = reshape(A', 1, [])};$
${\rm P = eye(M^2)};$
${\rm P = P(v,:)}$.} ${\bf P}$, which transforms $\vec(\Thetab) = \r$ into $\vec(\Thetab^T) = \r$,
\begin{equation}
\r = {\bf P} \r. \label{eq:commutationmatrix}
\end{equation}
This expression implies that the vectorization of any symmetric BD-RIS matrix belongs to the nullity of the matrix $\I_{M^2}-{\bf P}$. The rank of $\I_{M^2}-{\bf P}$ is $M(M-1)/2$ and the rank of its null space is $M^2 - M(M-1)/2 = M(M+1)/2$, which is the dimension of the space of symmetric $M\times M$ matrices. Let ${\bf N} \in \mathbb{C}^{M^2 \times (M(M+1)/2) }$ be a basis for the nullity of $\I_{M^2}-{\bf P}$. Then, any symmetric solution of the min-IL problem is of the form
\begin{equation}
    \r = {\bf N} \x,
\end{equation}
where $\x \in \mathbb{C}^{(M(M+1)/2) \times 1}$. Therefore, the relaxed min-IL problem \eqref{eq:minILbisa} can be rewritten as the following QCQP problem
\begin{align}
\label{eq:BDRISgroupr}
({\cal P}_2): \,\min_{\x}\,\,& \x^H {\bf N}^H \Sigmab {\bf N} \x + 2 {\rm Re}(\x^H {\bf N}^H{\bf s}) \\
\text{s.t.}\,\,& \x^H \x \leq M, \label{eq:powerBDRISa}
\end{align}
whose solution is
\begin{equation}
    \x = -({\bf N}^H \Sigmab {\bf N} +\lambda \I_{d_m})^{-1} {\bf N}^H{\bf s},
    \label{eq:solBDRIS}
\end{equation}
where $d_m = M(M+1)/2$ and $\lambda \geq 0$ is a parameter selected to meet the power constraint. This regularization parameter can be found by bisection. Alternatively, we could use CVX to solve ${\cal P}_2$. Once the solution for $\x$ has been found, the optimal value of $\r$ is  $\r = {\bf N} \x$. The BD-RIS obtained as $\Thetab = \unvec(\r)$ is symmetric but not unitary. The final solution is obtained by projecting $\Thetab$ into the set of unitary matrices as described in \cite{MaoCL2024,ZhouCL2025} and repeated here for completeness. Compute the SVD of the symmetric solution as $\Thetab = \U \Lambdab \V^H$, and partition the $\U$ and $\V$ matrices as $\U = [ \U_{d}, \U_{M-d}]$ and $\V = [\V_{d}, \V_{M-d}]$, where $d$ denotes the rank of $\Thetab$. The unitary and symmetric solution is ${\rm uni}(\Thetab) = [\U_{d}, \V_{M-d}^*]\V^H$.

\begin{remark}
An equivalent way to derive the optimal unitary projection is through Takagi's factorization \cite{Takagi}. The symmetric matrix $\Thetab$ can be factorized as $\Thetab = \Q \Lambdab \Q^T$, where $\Q$ is unitary and $\Lambdab$ is diagonal. Then, the optimal unitary projection is ${\rm uni}(\Thetab) = \Q \Q^T$.
\end{remark}

A summary of the RtP BD-RIS algorithm is shown in Algorithm \ref{alg:BDRISrelaxed}.

\begin{algorithm}[!t]
\small
\DontPrintSemicolon
\SetAlgoVlined
\KwIn{Initial $\H_{lk}, \F_k, \G_l$, $\forall k \neq l$}
\KwOut{Final BD-RIS $\Thetab$}{
Calculate $\s = \vec \left( \sum_{l \neq k} \F_k^H \H_{lk}\G_l \right)$, and $\Sigmab = \sum_{l \neq k}  \G_{l}^T \G_{l}^* \otimes 
 \F_{k}^H\F_{k} $  \;
Calculate the $M^2 \times M^2$ commutation matrix ${\bf P}$ \;
Calculate a basis ${\bf N}$ for the nullity of  $\I_{M^2}-{\bf P}$ \;
Solve problem ${\cal P}_2$ to obtain $\Thetab = \unvec({\bf N}\x)$ \;
Compute ${\rm uni}(\Thetab)$ projecting $\Thetab$ into the set of unitary matrices}
\caption{{\small RtP MinIL BD-RIS}}
\label{alg:BDRISrelaxed}
\end{algorithm}

\paragraph{Computational complexity}
To get the relaxed symmetric BD-RIS it is necessary to solve a QCQP problem of size $M(M+1)/2$, whose complexity is $\mathcal{O}((M(M+1)/2)^{3.5}$) \cite{LuoSPMag2010}, and a SVD of an $M \times M$ matrix for the unitary projection, with complexity $\mathcal{O}(M^3)$. Therefore, the total complexity of the RtP method is $\left(\mathcal{O}((M(M+1)/2)^{3.5}) + \mathcal{O}(M^3) \right)$, which is dominated by the first term.

\subsection{Group-connected BD-RIS}
\label{sec:GroupConnected}
To reduce the computational and circuit complexity of the fully-connected BD-RIS architecture, a group-connected architecture was proposed in \cite{ClerckxTWC22a, ClerckxTWC22b}. The $M$ reflective elements of a group-connected BD-RIS are partitioned into $G$ groups, each of $M_g = M/G$ elements. The elements of each group are fully connected but disconnected from the other groups. For a group-connected architecture, the min-IL optimization problem is
\begin{subequations}
\begin{align}
({\cal P}_{1b}): \,\min_{\Thetab}\,\,& \sum_{l \neq k} \| \H_{lk} + \F_k \Thetab \G_l^H \|_F^2 \label{eq:minILgroup}  \\
\text{s.t.}\,\,& \Thetab = \blkdiag(\Thetab_1,\ldots, \Thetab_G),\\
            \,\,& \Thetab_g^T = \Thetab_g, \, \Thetab_g^H \Thetab_g = \I_{M_g}.
\end{align}
\end{subequations}

If the direct channel is blocked it is easy to check that the cost function \eqref{eq:minILgroup} decouples into $G$ independent problems, each involving an $M_g \times M_g$ BD-RIS matrix. However, when the direct channel is not blocked the problem remains coupled among blocks, which makes it difficult to find a non-iterative solution. 

The proposed procedure follows an iterative approach in which the blocks are optimized one at a time. For the optimization of the $r$th block $\Thetab_r$, the blocks $g = 1,\ldots, r-1, r+1, \ldots, G$ are considered to be fixed. This procedure decouples the optimization of the different blocks, which are now solved individually. The optimization problem for the $r$th block, $\Thetab_r$, is
\begin{subequations}
\begin{align}
 \,\min_{\Thetab_r} \,\,& \sum_{l \neq k} \| \C_{kl} + \F_{kr} \Thetab_r \G_{lr}^H \|_F^2 \label{eq:minILgroupbis}  \\
\text{s.t.} \,\, & \Thetab_r^T = \Thetab_r, \, \tr(\Thetab_r^H \Thetab_r) \leq M_g,           
\end{align}
\end{subequations}
\noindent where $\C_{kl} = \H_{lk} + \sum_{g=1, g\neq r}^{R} \F_{kg} \Thetab_g \G_{lg}^H$ denotes the fixed part of the equivalent channel, and $\F_{kg}, \G_{lg}$ are, respectively, $N_{R_k} \times M_g$ and $N_{T_l} \times M_g$ matrices formed by extracting the $M_g$ columns of $\F_{k}$ and $\G_{l}$ acting on the $r$th BD-RIS block. To solve Problem (17), we can use either the MO algorithm described in Subsection \ref{sec:plBDRISalg} (cf. Algorithm \ref{alg:BDRISopt}), or the RtP version described in Subsection \ref{sec:relaxed} (cf. Algorithm \ref{alg:BDRISrelaxed}). 
The procedure is initialized to a feasible unitary and symmetric solution such as $\Thetab = \I_M$. All groups are sequentially optimized at each iteration and the optimization stops when the IL difference in two successive iterations falls below a threshold. Note that the IL decreases or remains constant after optimizing each group; therefore, the convergence of the algorithm is guaranteed. A summary of the iterative algorithm for group-connected BD-RIS is shown in Algorithm 3.

\begin{algorithm}[!t]
\small
\DontPrintSemicolon
\SetAlgoVlined
\KwIn{Initial $\H_{lk}, \F_k, \G_l$, $\forall k \neq l$, $\Thetab_g = \I_{M_g}$, $g=1,\ldots, G$}
\KwOut{Final BD-RIS $\Thetab = \blkdiag(\Thetab_1,\ldots,\Thetab_G)$}
\While{Convergence criterion not true}
{\For{$r =1,\ldots, G$}{
Update $\Thetab_r$ with either the MO algorithm or the RtP algorithm \;
}}
\caption{{\small Group-connected MinIL BD-RIS}}
\label{alg:BDRISgroup}
\end{algorithm}

\paragraph{Computational complexity}
Taking the optimization of each group by the RtP method as an example, the computational cost per iteration to optimize a group-connected BD-RIS with $G$ groups of $M_g$ elements each is $G \left(\mathcal{O}((M(M+1)/2 \right)^{3.5})$. The number of iterations to reach convergence is usually small (less than 20 iterations in all the simulations described in Section \ref{sec:simulations}). As $M_g << M$ the complexity to optimize a group-connected architecture can be significantly lower than that of a fully-connected BD-RIS.

\subsection{Min-IL passive lossless RIS}
\label{sec:RIS}
It may be of interest to compare the results obtained for a BD-RIS with those of a diagonal RIS $\bm{\Theta}=\text{diag} \left(\r \right)$ where $ \r = \left(r_1, r_2,\cdots,r_M\right)^T$ and $r_m$ is the $m$th reflecting element. The IL cost function for a RIS has the same quadratic expression as for a BD-RIS
\begin{equation*}
IL(\r) = \tr(\T) + \r^H \Sigmab \r + 2 {\rm Re}(\r^H {\bf s}),
\end{equation*}
where now $ \r = \diag(\Thetab)$, $\Sigmab = \sum_{l \neq k}  \G_l^T \G_l^* \odot 
 \F_k^H\F_k$ and $ \s = \vec \left( \sum_{l \neq k} \F_k^H \H_{lk}\G_l \right)$.
Interestingly, the IL minimization problems for RIS and BD-RIS are structurally identical. Note however that the dimensions of the quadratic problem, and hence the computational complexity to solve it, are of the order of $M$ for an RIS and of the order of $M^2$ for a fully-connected BD-RIS.

The minimization problem is
\begin{subequations}
\begin{align}
\label{eq:RISpl}
({\cal P}_2): \,\min_{\r}\,\,&  \r^H \Sigmab \r + 2 {\rm Re}(\r^H {\bf s}) \\
\text{s.t.}\,\,& |r_m| =1, \forall m.
\end{align}
\end{subequations}
\noindent This is a unit-modulus quadratic programming problem. The function to be minimized is convex but the constraint is not. There are several iterative algorithms in the literature to solve ${\cal P}_2$ \cite{SidiropoulosTSP2012},\cite{pan2020multicell},\cite{tsinos2017efficient}, \cite{SantamariaICASSP23}. The block coordinate descent (BCD) algorithm in \cite{SantamariaICASSP23} is used in this paper. At each step of the BCD algorithm, all values of $\r$ but the $m$th $r_m = e^{j \theta_m}$ are fixed. Denoting $\overline{m} = \{ 1, \ldots, m-1,m+1,\ldots, M \}$, the IL as a function of $r_m$ can be written as
\begin{equation}
IL(r_m) = \beta + 2 {\rm Re}(r_m^*(s_m + \Sigmab_{\overline{m}}^H \r_{\overline{m}})), 
\label{eq:pBCD}
\end{equation}
where $\beta$ is a constant, $\Sigmab_{\overline{m}}$ denotes the $m$th column of $\Sigmab$ with the $m$th element removed, and 
\[\r_{\overline{m}} = (r_1, \ldots, r_{m-1}, r_{m+1}, \ldots, r_M).
\]
The optimization problem for the $m$th RIS coefficient is then
\begin{align}\label{eq_RIS_UMAO}
\min_{r_m}\,\,&  {\rm Re}(r_m^*(s_m + \Sigmab_{\overline{m}}^H \r_{\overline{m}})) \\
\text{s.t.}\,\,& |z_m| = 1,
\end{align}
which has the following closed-form solution
\begin{equation}
\theta_m^* = \angle{\left(s_m + \Sigmab_{\overline{m}}^H \r_{\overline{m}} \right)} - \pi.
\label{eq:RISplupdate}
\end{equation}

\section{Stage II: Precoder design}
\label{sec:StageII}
Once the BD-RIS has been optimized to minimize the IL, the optimization of the precoders and decoders can be performed according to any of the criteria already considered in the classical MIMO literature. This section summarizes the selected criteria and the algorithms used. Although the optimization problems in this section have been studied for the MIMO-IC without BD-RIS, it is convenient to provide a summary of them for the completeness of the paper, as well as to understand the complexity analyses and simulation results in Sec. VI. To simplify the notation in this section we consider all users to have $N \times N$ symmetric MIMO links. In addition, remember that the equivalent channels of the BD-RIS-assisted MIMO IC after Stage I are $\widetilde{\H}_{lk}$ (see \eqref{eq:Eqchannel}). 

\paragraph{\bf Interference oblivious SVD-based precoders \cite{BlumJSAC03}} 
Each user independently optimizes its precoder by considering only the receiver noise and disregarding the residual interference not eliminated by the BD-RIS. If the BD-RIS designed in Stage I reduces the interference significantly below the noise level, the MIMO-IC channel decouples into $K$ parallel MIMO Gaussian channels and the SVD-based precoders are capacity-achieving \cite{BlumJSAC03}. This typically requires a very large number of BD-RIS elements. There will always be residual interference with a more limited number of elements, therefore the SVD-based solution may be rather suboptimal. For a fixed BD-RIS, the interference oblivious SVD-based transmit covariance matrix for user $k$ is $\R_{xx,k} = \V_k {\bf P}_k \V_k^H$, where $\V_k$ is the left eigenspace of the equivalent MIMO channel $\widetilde{\H}_{kk}$ and ${\bf P} = \diag(p_1,\ldots,p_d)$ is a diagonal matrix where $p_i$ denotes the optimal power allocated to the $i$th stream given by the water-filling strategy to satisfy the power constraint $\sum_i p_i = P_{t_k}$.

\paragraph{\bf Min-IL precoders \cite{gomadam2011distributed, Gonzalez14}}
The precoders and decoders are designed to minimize the residual IL that the BD-RIS may not have eliminated. This approach can be considered as an altruistic design of the precoders in which each user tries to generate the least interference to other users. In contrast, SVD-based precoders can be considered a selfish design in which each user tries to maximize the rate over its direct channel without considering the interference it causes to (or receives from) other users. The optimization problem is to find precoders, $\V_k$, and decoders, $\U_k$, to minimize
\begin{equation}
 IL(\{\V_k\},\{\U_k\}) = \sum_{l \neq k} \| \U_k^H \widetilde{\H}_{lk} \V_l \|_F^2.
\label{eq:ILminprec}
\end{equation}
This problem can be solved by applying the AO method described in \cite{gomadam2011distributed} that optimizes the decoders for a fixed set of precoders and then the precoders for a fixed set of decoders. Each step has a closed-form solution. Let us take the decoder of the $k$th user as an example. The $N \times N$ covariance matrix of the interference at the $k$th receiver is
\begin{equation*}
    {\bf S}_{k} = \sum_{l=1,l \neq k}^K  {\bf S}_{kl}= \sum_{l=1,l \neq k}^K \widetilde{\H}_{lk} \V_l \V_l^H \widetilde{\H}_{lk}^H, 
\end{equation*}
where $\widetilde{\H}_{kl}$ is the equivalent MIMO channel in \eqref{eq:Eqchannel} and ${\bf S}_{kl} = \widetilde{\H}_{lk} \V_l \V_l^H \widetilde{\H}_{lk}^H$ is the interference covariance matrix from transmitter $l$ to receiver $k$. The decoder $\U_k$ is thus obtained as the eigenvectors of ${\bf S}_{k}$ corresponding to their $d_k$ smallest eigenvalues. To design the min-IL precoders, the channel reciprocity is exploited by exchanging the roles of transmitters and receivers, and transmitting along the directions where each transmitter generates less interference. Let us recall that \cite{ChannelRecWC23} shows that there is still channel reciprocity for RIS-assisted wireless networks unless the RIS uses active nonreciprocal
circuits or has structural asymmetries \cite{SwindlehurstJOC23}. We will refer to the precoders designed in this way as min-IL precoders.

\paragraph{\bf Max-SINR precoders \cite{gomadam2011distributed, Wilson2013MaxSINR}} The min-IL precoders designed to minimize \eqref{eq:ILminprec} do not consider the direct channels. An alternative approach that obtains a better sum rate than min-IL precoders designs the precoders and decoders to maximize the SINR \cite{gomadam2011distributed, Wilson2013MaxSINR}. The max-SINR algorithm follows an AO procedure similar to the min-IL precoders. Details can be found in \cite{gomadam2011distributed, Wilson2013MaxSINR}.

\paragraph{\bf Max sum-rate precoders \cite{soleymani2020improper, soleymani2024optimization}} 
If interference is sufficiently reduced by the BD-RIS designed in Stage I, treating interference as noise (TIN) is the optimal decoding strategy to maximize the sum rate \cite{annapureddy2009gaussian}. Utilizing TIN, the achievable rate of user $k$ is 
\begin{equation}   \label{eq-rate}
    r_{k}\!=\!
\log \det \left( \I_N + \left(\sigma^2 \I_N + \sum_{l=1,l\neq k}^K {\bf S}_{kl} \right)^{-1} {\bf S}_{kk} \right),\!
\end{equation}
where ${\bf S}_{kl} = \widetilde{\H}_{lk} \V_l \V_l^H \widetilde{\H}_{lk}^H$. Thus, the sum rate maximization problem can be written as

\begin{align}\label{eq-sum-rate}
 \underset{\{\V_1, \ldots, \V_K\}
 }{\max}\,\,\,  & 
  \sum_k r_k&
  \text{s.t.}\,\,   \,&  \tr \left({\bf V}_{k}{\bf V}_{k}^H\right)\leq P_{t_k},\,\forall k,
 \end{align}
which is non-convex. To solve \eqref{eq-sum-rate}, we employ the majorization-minimization (MM) iterative approach, proposed in \cite{soleymani2020improper, soleymani2024optimization}. To this end, we first obtain surrogate functions for the rates by leveraging the inequality in the following lemma.
\begin{lemma}[\!\cite{soleymani2022improper}]\label{lem-3} 
Consider arbitrary $N \times N$ matrices ${\bf \Lambda}$ and $\bar{{\bf \Lambda}}$, and $N \times N$ positive definite matrices ${\bf \Upsilon}$ and $\bar{{\bf \Upsilon}}$. Then, we have:
\begin{multline} 
\ln \det \left(\I_N+{\bf \Upsilon}^{-1}{\bf \Lambda}{\bf \Lambda}^H\right)\geq
 \ln \det \left(\I_N+{\bf \Upsilon}^{-1}\bar{{\bf \Lambda}}\bar{{\bf \Lambda}}^H\right)
\\-
\tr\left(
\bar{{\bf \Upsilon}}^{-1}
\bar{{\bf \Lambda}}\bar{{\bf \Lambda}}^H
\right)
+
2\mathfrak{R}\left\{\tr\left(
\bar{{\bf \Upsilon}}^{-1}
\bar{{\bf \Lambda}}{\bf \Lambda}^H
\right)\right\}\\
-
\tr\left(
(\bar{{\bf \Upsilon}}^{-1}-(\bar{{\bf \Lambda}}\bar{{\bf \Lambda}}^H + \bar{{\bf \Upsilon}})^{-1})^H({\bf \Lambda}{\bf \Lambda}^H+{\bf \Upsilon})
\right),
\label{lower-bound}
\end{multline}
where $\mathfrak{R}\{a\}$ denotes the real part of $a$.
\end{lemma}
 The lower bound in Lemma \ref{lem-3} allows us to obtain a concave lower bound for $r_k$ as presented in the following lemma.
\begin{lemma}[\!\cite{soleymani2024optimization}]\label{lem-4}
A concave lower bound for $r_{k}$ is given by 
\begin{multline}
\label{eq24}
r_{k}\geq \bar{r}_{k}= a_{k}
+2\sum_{l}\mathfrak{R}\left\{\tr\left(
{\bf A}_{lk}\mathbf{V}_{l}^H
\tilde{\mathbf{H}}_{lk}^H\right)\right\}
\\
-
\tr\left(
{\bf B}_{k}\left(\sigma^2 \I_N+\sum_{l=1}^K \widetilde{\H}_{lk} \V_l \V_l^H \widetilde{\H}_{lk}^H\right)
\right),
\end{multline}
where 
\begin{align*}
a_{k}&=\ln \det \left (\I_N +\mathbf{R}_{k}(\bar{\V}_l)\widetilde{\H}_{lk} \bar{\V}_k \bar{\V}_k^H \widetilde{\H}_{lk}^H\right)
\\&\hspace{.8cm}
-
\tr\left(
\mathbf{R}_{k}(\bar{\V}_l)\widetilde{\H}_{lk} \bar{\V}_k \bar{\V}_k^H \widetilde{\H}_{lk}^H
\right),
\\
{\bf A}_{lk}&= 
\mathbf{R}_{k}(\bar{\V}_l)
\tilde{\mathbf{H}}_{kk}\bar{\mathbf{V}}_{k},
\\
{\bf B}_{k}&=\mathbf{R}_{k}(\bar{\V}_l)
-\left(\sigma^2 \I_N +\sum_{l=1}^K \mathbf{S}_{k}(\bar{\V}_l)\right)^{-1},
\\
\mathbf{R}_{k}(\bar{\V}_l) &=  \left(\sigma^2 \I_N+\sum_{l=1,l\neq k}^K  \widetilde{\H}_{lk} \bar{\V}_l \bar{\V}_l^H \widetilde{\H}_{lk}^H\right)^{-1}, 
\\
\mathbf{S}_{k}(\bar{\V}_l) &=  \widetilde{\H}_{lk} \bar{\V}_l \bar{\V}_l^H \widetilde{\H}_{lk}^H,
\end{align*} 
and $\bar{\mathbf{V}}_{l}$ is the initial value of ${\mathbf{V}}_{l}$ for all $l=1,\cdots,K$. 
\end{lemma}
Substituting $r_k$ by $\bar{r}_k$ yields the convex problem
\begin{align}\label{eq-sum-rate2}
 \underset{\{\V_1, \ldots, \V_K\}
 }{\max}\,\,\,  & 
  \sum_k \bar{r}_k&
  \text{s.t.}\,\,   \,&  \text{Tr} \left({\bf V}_{k}{\bf V}_{k}^H\right)\leq P_{t_k},\,\forall k.
 \end{align}
Iteratively solving \eqref{eq-sum-rate2} and updating $\{\bar{\bf V}_1,\ldots, \bar{\bf V}_K \}$ converges to a stationary point of \eqref{eq-sum-rate}.

\begin{remark}
    Note that the SVD-based and max-SR precoders optimize the number of transmitted streams. However, this number is fixed and predetermined, $\sum_{k=1}^K d_k$, for the min-IL and max-SINR precoders.
\end{remark}

\paragraph{Computational complexity}
Computing the SVD of the direct MIMO links dominates the computational complexity of designing SVD-based precoders and therefore is cubic in the number of antennas $\mathcal{O}(K N^3)$. The design of min-IL or max-SINR precoders using AO involves $\mathcal{O}(2K N^3)$ floating-point operations at each AO iteration due to the computation of SVDs for optimizing the $K$ precoders and $K$ decoders. Finally, each iteration of the MM procedure to design the max-SR precoders involves an outer iteration to compute the surrogate function for the rates, whose computational complexity is dominated by the matrix inversion in Lemma \ref{lem-4}, $K\mathcal{O}(N^3)$, and an inner iteration in which a convex problem has to be solved. The number of Newton iterations to solve a convex problem is proportional to the square root of the number of constraints \cite{boyd2004convex}, and hence it grows with $\sqrt{K}$. To solve at each Newton step, the calculation of the surrogate function is $\mathcal{O}((K+1)N^2)$. In summary, the computational complexity for designing max-SR precoders per MM iteration is $\mathcal{O}(K N^3 + \sqrt{K}(K+1)N^2)$.

\subsection{Joint min-IL precoder and BD-RIS design}
\label{sec:AO}

The proposed two-stage optimization approach has clear computational advantages over the joint optimization of the BD-RIS and the precoders. However, it is relatively straightforward to implement an AO algorithm that jointly optimizes the BD-RIS and the precoders. As an example, we consider the joint optimization of precoders, decoders, and BD-RIS to minimize the IL cost function 
\begin{equation}
IL(\{\V_k\},\{\U_k\},\Thetab) = \sum_{l \neq k} \|\U_k^H \H_{lk} \V_l + \U_k^H \F_k \Thetab \G_l^H \V_l  \|_F^2.
\label{eq:ILcostfunctionbis}
\end{equation}
The minimization of \eqref{eq:ILcostfunctionbis} may be carried out through a 3-step AO process, where in each step a set of variables (decoders, precoders, or BD-RIS) is optimized while the other variables are held fixed:
\begin{enumerate}
    \item Optimize $\{\U_k\}_{k=1}^K$ for fixed $(\{\V_l\}_{l=1}^K, \Thetab)$.
    \item Optimize $\{\V_l\}_{l=1}^K$ for fixed $(\{\U_k\}_{k=1}^K, \Thetab)$.
    \item Optimize $\Thetab$ for fixed $(\{\U_k\}_{k=1}^K, \{\V_l\}_{l=1}^K)$.
\end{enumerate}
Steps 1 and 2 can be solved by applying the min-IL design of the precoders described in Sec. \ref{sec:StageII}, while Step 3 can be solved with the MO algorithm of Sec. \ref{sec:StageI}. Algorithm \ref{alg:BDRISoptjoint} details the joint min-IL optimization procedure.
\begin{algorithm}[!t]
\small
\DontPrintSemicolon
\SetAlgoVlined
\KwIn{Channels, initial $\{\U_k\}_{k=1}^K$ and $\{ \V_k \}_{k=1}^K$; $IL_{NoRIS}$ (IL without RIS), IL convergence threshold $\epsilon >0$}
\KwOut{Final BD-RIS $\Thetab$,  $\U_k$ and $\V_k$ for $k=1,\ldots,K$ }
\While{Relative IL reduction smaller than $ \epsilon$}{
\tcc{Step 1: Given $\Thetab$ and $\{\V_k \}_{k=1}^K$ }
Calculate ${\bf S}_{k}$ at the $k$th receiver and obtain $\U_k$ from the eigenvectors corresponding to its smallest $d_k$ eigenvalues $k=1,\ldots, K$\;
\tcc{Step 2: Given $\Thetab$ and $\{\V_k \}_{k=1}^K$ }
Calculate ${\bf S}_{k}$ at the $k$th transmitter and obtain $\V_k$ from the eigenvectors corresponding to its smallest $d_k$ eigenvalues $k=1,\ldots, K$\;
\tcc{Step 3: Given $\{\V_k \}_{k=1}^K$ and $\{\U_k \}_{k=1}^K$ update $\Thetab$ with the MO algorithm described in Algorithm \ref{alg:BDRISopt}}
}
\caption{{\small Joint min-IL design}}
\label{alg:BDRISoptjoint}
\end{algorithm}

\section{Simulation results}
\label{sec:simulations}

\subsection{Scenario description}
We consider a $(3 \times 3, 2)^3$ MIMO-IC assisted by a BD-RIS with $M$ elements as represented in Fig. \ref{fig:setup}. For simplicity, all the MIMO links have the same number of antennas, $3 \times 3$, and transmit the same number of streams, $d =2$. In this scenario, IA is not feasible by optimizing only the precoders and decoders \cite{gonzalez2014feasibility}. The system is strongly limited by interference and it is necessary to assist the channel with a BD-RIS optimized to minimize the interference. For this scenario, Lemma \ref{Prop:LemmaUnc} states that an unconstrained (i.e., active) BD-RIS with $M = 8$ elements could achieve $IL=0$, since in this case $\Sigmab$ in \eqref{eq:ILcostfunction1} is a $64 \times 64$ matrix with rank $ \min(M^2,\sum_{k \neq l} N_{R_k} N_{T_l} ) = 54$. Of course, a passive BD-RIS design will require a much larger number of elements to achieve a significant interference reduction level.

The $K =3$ transmitters and receivers are regularly located in a square of $ 50$ m. The coordinates $(x,y,z)$ in meters $[m]$ of the $k$th transmitter and the $k$th receiver are (0,\,$50 \frac{(k-1)}{(K-1)}$,\,1.5) and ($50$,\,$50 \frac{(k-1)}{(K-1)}$,\,1.5) $k=1,\ldots, K$, respectively. The BD-RIS located at ($x$,\, $y$,\,5) $[m]$, where the $(x,y)$ coordinates can vary between 5 and 45 m. Since the BD-RIS is located higher than the transmitters and receivers with a direct line-of-sight (LoS) to all of them, it is assumed that all transmitters and receivers are in the reflection space of the RIS. The large-scale path loss in dB  is given by
\begin{equation*}
PL = -28 - 10 \alpha \log_{10} \left( r \right),
\label{eq:pathloss}
\end{equation*}
where $r$ is the link distance, and $\alpha$ is the path-loss exponent. The direct Tx-Rx links, $\H_{lk}$, are assumed to be non-line-of-sight (NLoS) channels, with path-loss exponent $\alpha = 3.75$ and small-scale Rayleigh fading. The Tx-RIS-Rx links, $\F_k$ and $\G_l$, are assumed to be LoS channels, with path-loss exponent $\alpha = 2$ and small-scale Rice fading with Rician factor $\gamma =3$. The noise variance is calculated as $ \sigma^2 = -174 ({\rm dBm/Hz}) + 10\log_{10} B (\rm Hz) + F ({\rm dB})$, where $B$ is the bandwidth that we take in our simulations as $B= 40$ MHz and $F$ is the noise figure of the receiver that we take as $F = 10$ dB. A new set of channels is generated in each simulation, but the positions of Txs, Rxs, and BD-RIS are fixed. We average the results of 50 independent simulations. The power transmitted by all users is $P_t = 10$ dBm. 

As a figure of merit, we use the difference in dBs between the IL ratios with and without BD-RIS, which is also the difference in dBs between the INR with and without BD-RIS
\begin{equation}
\small
    \Delta{\text{INR}} = 10 \log_{10} \left ( \frac{\sum_{l \neq k} \| ({\H}_{lk} + {\F}_k \Thetab {\G}_l^H )\|_F^2}{\sum_{l \neq k} \| {\H}_{lk} \|_F^2} \right).
    \label{eq:INR}
\end{equation}
Another figure of merit used in the simulations is the sum rate of the MIMO IC
\begin{equation}
    R = \sum_{k=1}^K r_k,
    \label{eq:sumrate}
\end{equation}
where $r_k$ is the rate of user $k$ given by \eqref{eq-rate}. 

\begin{figure}
    \centering
\includegraphics[width=.5\textwidth]{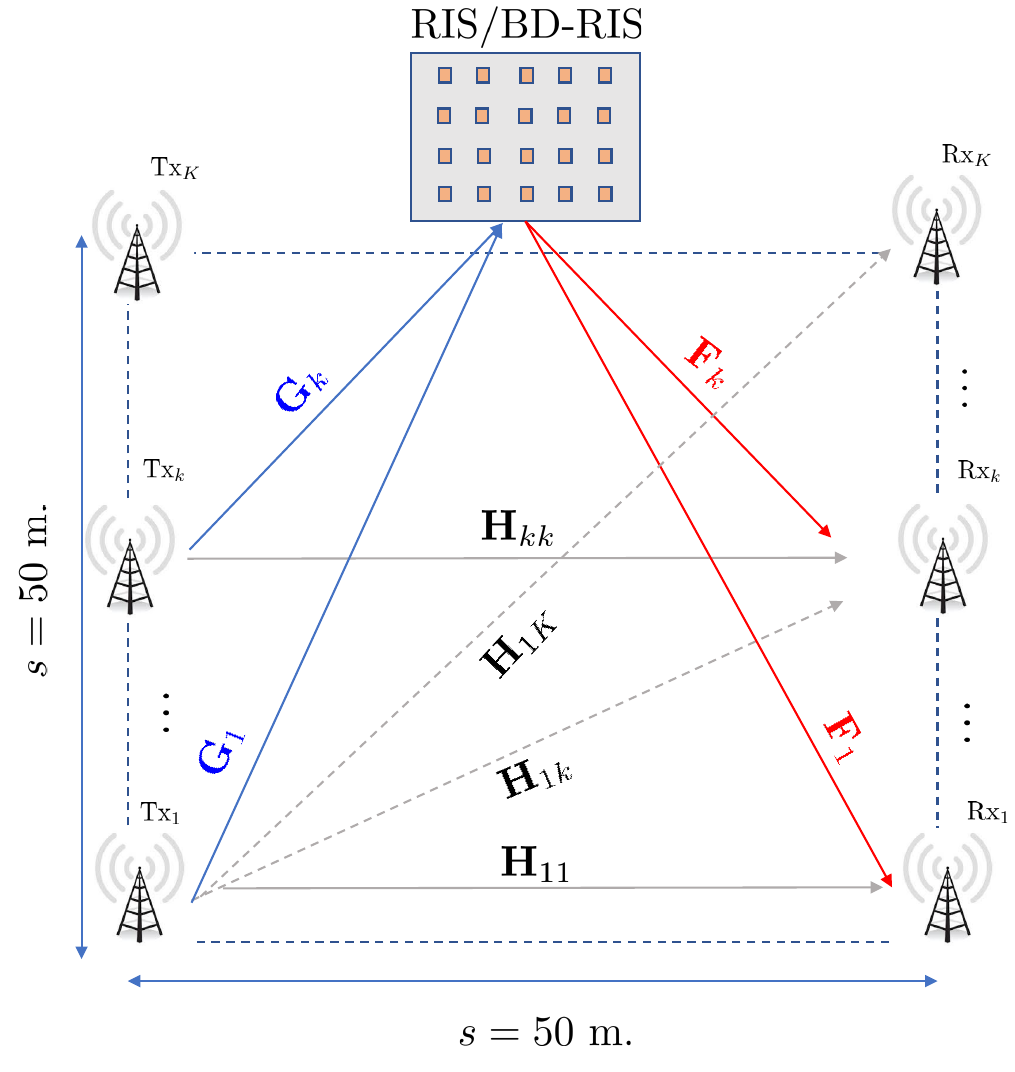}
     \caption{Simulation setup for the RIS-assisted MIMO-IC.}
	\label{fig:setup}
\end{figure}

\subsection{Optimal BD-RIS deployment}
The position of the reconfigurable surface or its {\it area of influence} is an essential aspect in optimizing system performance, as shown in \cite{AlexandropoulosJWCN2023}, \cite{SchoberPoorTCOM2023}. Therefore, we first evaluate the reduction in interference measured as in \eqref{eq:INR} achieved by varying the position in the $(x,y)$ plane of a fully-connected BD-RIS with $M=40$ elements. The positions in the $(x,y)$ plane are varied on a grid of 5 [m] along each coordinate while the height remains fixed at $z=5$ [m]. For the optimization of the BD-RIS we use the MO algorithm described in Algorithm \ref{alg:BDRISopt}. Fig. \ref{fig:Scenario1_Position} shows the results obtained using a color map, where we can observe two optimal positions or hot spots for the BD-RIS at approximately $(10, \, 25, \, 5)$ [m] and $(40,\, 25, \, 5)$ [m], where a reduction of about 10 dB in the interference level is achieved with respect to a scenario without BD-RIS. Unless otherwise stated in the remaining experiments on Scenario 1, we will position the BD-RIS at coordinates $(40, \,25, \,5)$ [m].

\begin{figure}[t!]
\centering		
\includegraphics[width=.5\textwidth]{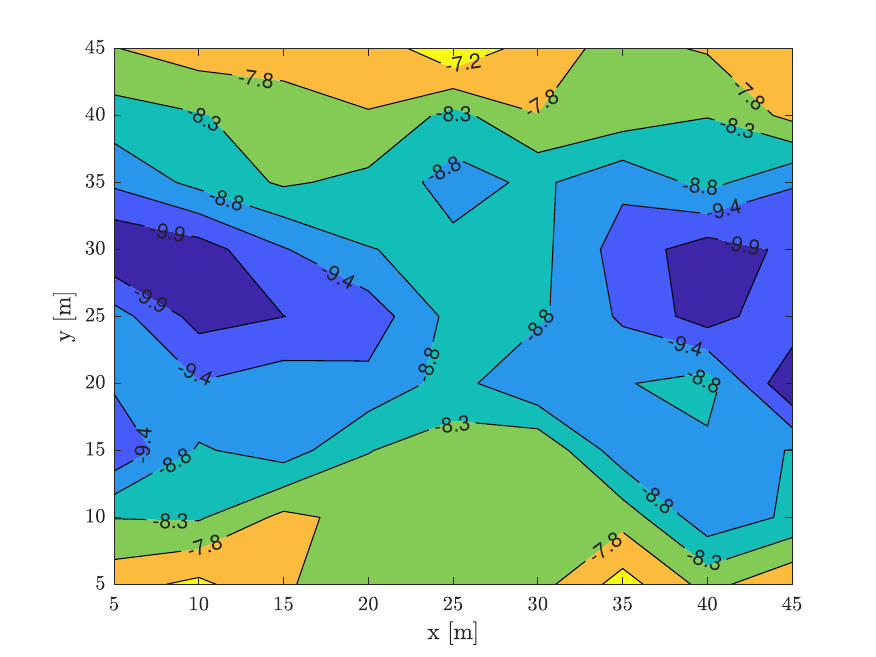}
\caption{$\Delta{\text{INR}}$ in dB as a function of the position of a BD-RIS with $M=40$ elements in a $(3 \times 3,2)^3$ MIMO-IC.} 
\label{fig:Scenario1_Position}	
\end{figure}

\subsection{MO and RtP convergence}
Fig. \ref{fig:ConvergenceMO} illustrates the convergence of the MO algorithm used to design a min-IL fully-connected BD-RIS with $M=40$ elements. Each curve shows the convergence for a different random initialization. The MO algorithm requires in this example 500 iterations to reach convergence. Remember that the complexity of each iteration is $\mathcal{O}(M^3)$, which may represent a difficulty when optimizing a BD-RIS with many elements. 
\begin{figure}[htp!]
\centering		
\includegraphics[width=.5\textwidth]{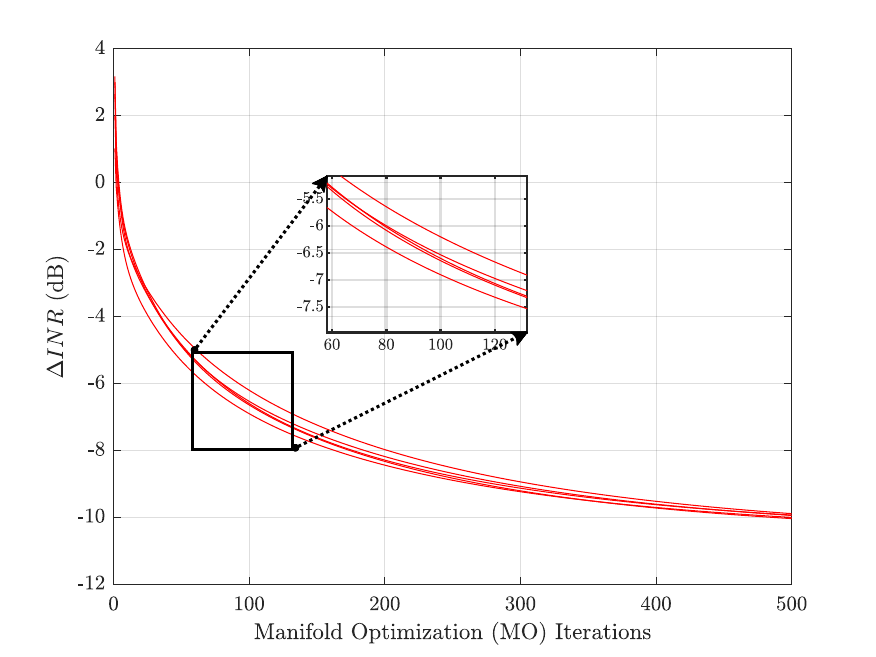}
\caption{Convergence of the MO iterative algorithm for a fully-connected BD-RIS with $M=40$ elements. Each curve shows the convergence for a different random initialization.} 
\label{fig:ConvergenceMO}	
\end{figure}
One way to reduce the computational and circuit complexities of the fully-connected BD-RIS is to employ a group-connected architecture. We study the computational complexity of Algorithm 3 when each group is optimized either by the MO algorithm or by the RtP method. As the number of iterations for the convergence of the two methods can be very different, we use the average run time in logarithmic scale as a figure of merit. The results for a group-connected BD-RIS with $M_g =4$ and $M_g=8$ for increasing values of $M$ are shown in Fig. \ref{fig:RunTimeGroupConnected} for the MO (solid line) and the RtP method (dashed line). The computational cost of the MO algorithm is several orders of magnitude higher than that of RtP. As expected, for both the MO and RtP the computational cost grows with $M_g$ and with the total number of elements $M$ and therefore the number of blocks to be optimized $G = M/M_g$. 

\begin{figure}[htp!]
\centering		
\includegraphics[width=.5\textwidth]{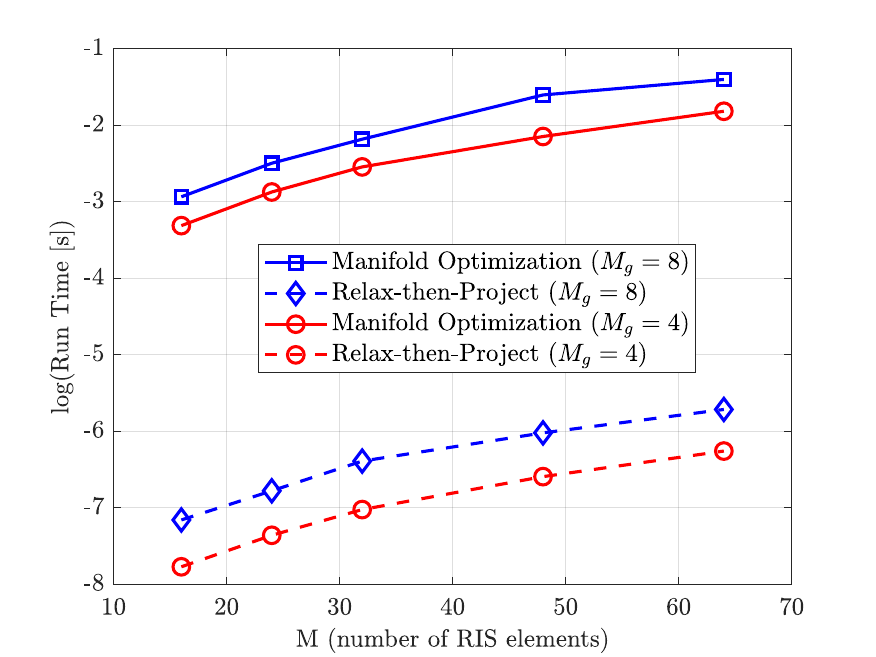}
\caption{Run time of the MO and RtP algorithms for the group-connected BD-RIS architecture vs. $M$ (total number of BD-RIS elements).} 
\label{fig:RunTimeGroupConnected}	
\end{figure}

\subsection{Fully-connected vs. group-connected BD-RIS}
In this section, we compare the interference suppression capability of the fully connected and group-connected architectures for various group sizes $M_g$. The BD-RIS is located at $(40, \,25, \,5)$ [m] and the interference suppresion capability for the different architectures and optimization algorithms is measured through the $\Delta{\text{INR}}$ defined in \eqref{eq:INR}. We employ either the MO algorithm or the RtP method to optimize the group-connected BD-RIS architecture. For completeness, we include the results obtained for a diagonal RIS, which can be considered a limiting case of the group-connected architecture with $M_g =1$. The results are shown in Fig. \ref{fig:Delta_INR_BDRIS_Scenario1}. For instance, an RIS with $M =128$ elements achieves a reduction in interference of about 7.5 dB compared to the no-RIS scenario. The interference suppression is improved by about 0.7 dB, 3 dB and 8 dB for group-connected BD-RIS of sizes $M_g =2,4,8$ optimized with the MO algorithm; and more than 11 dB when a fully-connected BD-RIS is employed. We observe that to achieve a significant interference reduction at the receivers it is necessary to employ a BD-RIS with many elements. The group-connected architecture optimized with the MO algorithm with a group size $M_g=8$ provides a reasonable tradeoff between performance and complexity. Moreover, compared to fully-connected architectures, group-connected architectures reduce not only the computational complexity of the optimization algorithms but also the number of interconnections between elements.

\begin{figure}[htp!]
\centering		
\includegraphics[width=.5\textwidth]{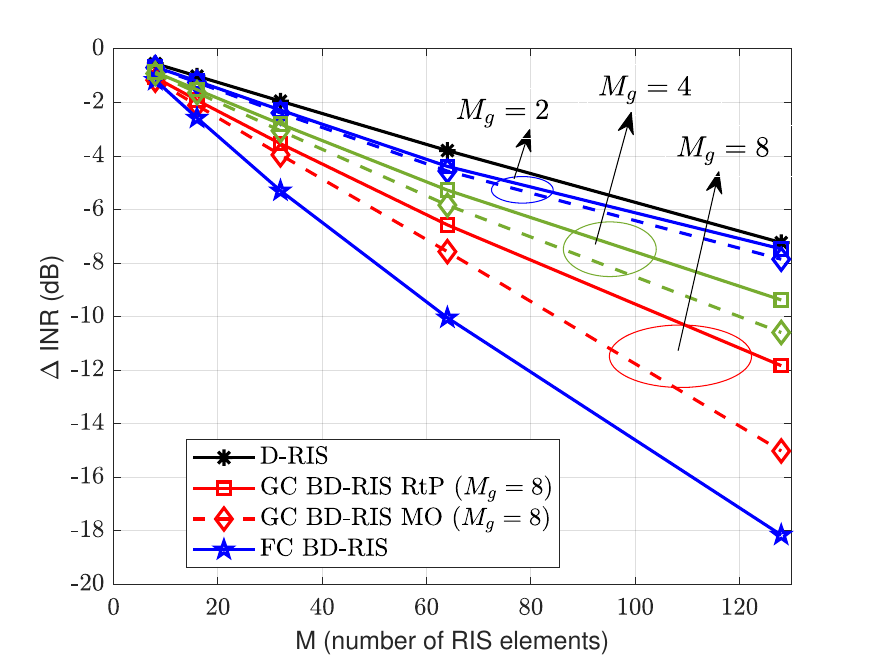}
\caption{$\Delta{\text{INR}}$ in dB as a function of $M$ for a fully-connected BD-RIS, a group-connected BD-RIS with group sizes $M_g= 2,4,8$ optimized with either MO or RtP, and a diagonal RIS.} 
\label{fig:Delta_INR_BDRIS_Scenario1}	
\end{figure}


\subsection{Sum rate performance}
In this subsection, we study the sum-rate performance of the $K$-user MIMO-IC assisted by a fully-connected BD-RIS. The BD-RIS is optimized in Stage I to minimize IL with the MO algorithm. In Stage II, we design precoders to optimize one of the following criteria: SVD-based, min-IL, max-SINR, or max sum-rate (labeled as max-SR in the figures). Additionally, we evaluate the joint min-IL design of the BD-RIS and the precoders as described in Algorithm \ref{alg:BDRISoptjoint}. Fig. \ref{fig:SRvsPower_Scenario1_300924} shows the sum rate vs. the transmit power curves.	The MIMO-IC is still interference-limited when the BD-RIS has only $M=16$ elements and the transmitted power is high. Therefore, the sum-rate curves for $M=16$ tend to saturate at high $P_t$. This is especially clear for the SVD-based precoders that do not take interference into account. However, a joint design of the BD-RIS and the precoders ensures that the MIMO-IC is limited mainly by noise, even at high $P_t$. When the number of BD-RIS elements increases to $M=64$, the BD-RIS is much more effective in reducing interference, and hence the differences between the different precoders are less pronounced. Among the different precoders, the max-SR design provides the best performance, with the drawback of a higher computational cost.

\begin{figure}[t!]
\centering		
\includegraphics[width=.5\textwidth]{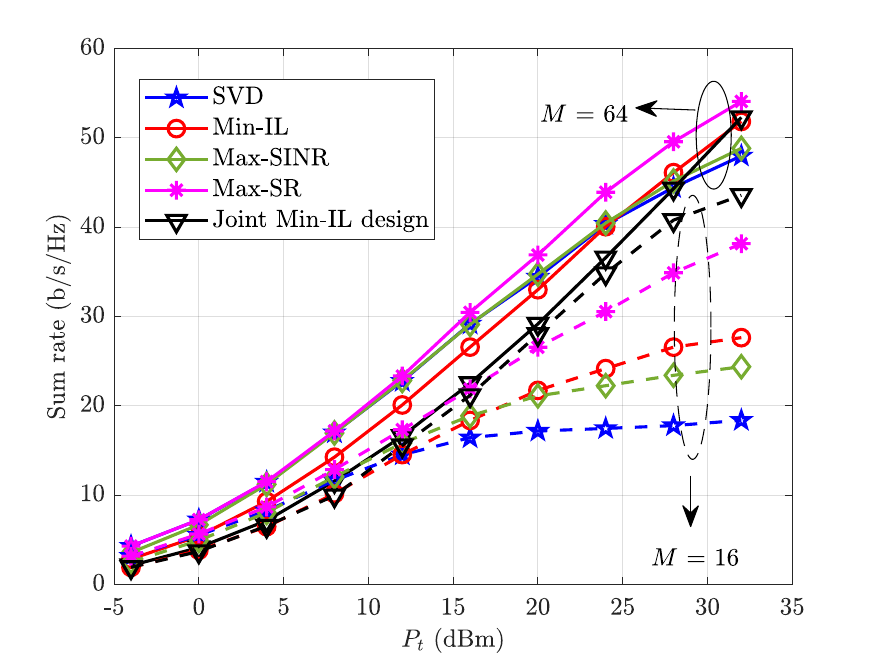}
\caption{Sum rate vs. $P_t$ for a BD-RIS assisted $(3 \times 3,2)^3$ MIMO-IC with different precoders. The BD-RIS has $M=16$ (dashed lines) or $M=64$ elements (solid lines).} 
\label{fig:SRvsPower_Scenario1_300924}	
\end{figure}

Fig. \ref{Fig/SRvsM_Scenario1_021024} shows the sum rate obtained by different precoders designed in Stage II as a function of the number of BD-RIS elements for two transmission powers. When the power transmitted by the users is $P_t = 10$ dBm, the interference level is not very high. Therefore, a BD-RIS with a moderate number of elements can eliminate much of this interference thus transforming the MIMO-IC into $K=3$ parallel MIMO Gaussian channels. This is why the differences between max-SR, max-SINR, and SVD-based precoders are negligible for $M \geq 60$. When the power transmitted by the users increases to $P_t = 20$ dBm, the residual interference not eliminated by the BD-RIS is higher and, therefore, the differences between the performance of the precoders are more remarkable. The max-SR precoders provide the best performance as expected. When the precoders and BD-RIS are jointly designed to minimize IL, a BD-RIS with $M=20$ elements is capable of drowning out interference below the noise level, and hence increasing $M$ does not translate into a sum rate improvement.

\begin{figure}[t!]
\centering		
\includegraphics[width=.5\textwidth]{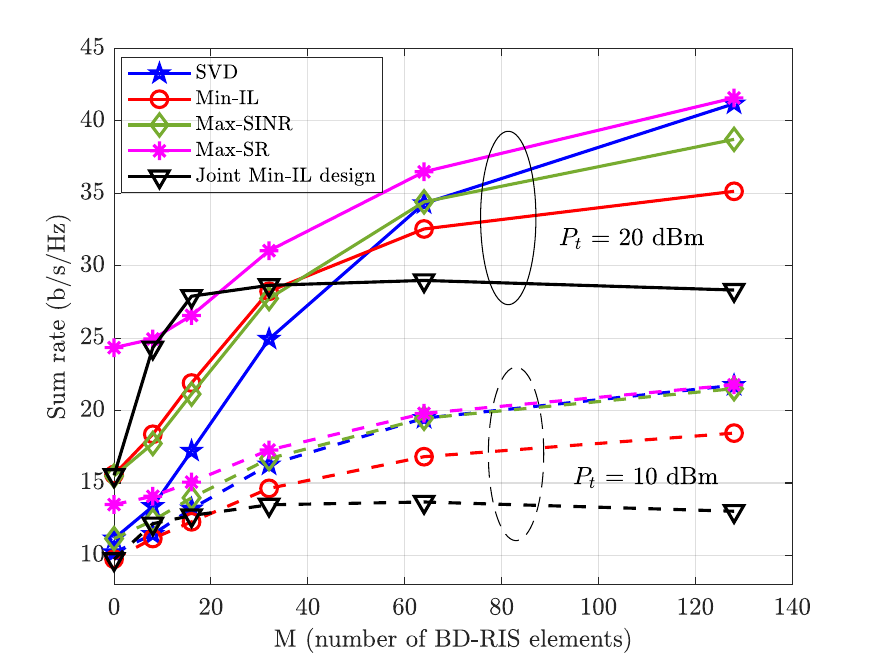}
\caption{Sum rate vs. $M$ for a BD-RIS assisted $(3 \times 3,2)^3$ MIMO-IC with different precoders. The users transmit $P_t=10$ dBm (dashed lines) or $P_t=20$ elements (solid lines).} 
\label{Fig/SRvsM_Scenario1_021024}	
\end{figure}



\section{Conclusions}
\label{sec:conclusions}
This work proposed a two-stage passive and active beamforming scheme for a BD-RIS-assisted $K$-user MIMO-IC. The passive BD-RIS, whose scattering matrix is unitary and symmetric, is designed in the first stage to minimize the IL. A geodesic MO algorithm on the manifold of unitary matrices, which exploits Takagi's factorization, is proposed to optimize a fully connected BD-RIS. For a group-connected BD-RIS architecture, a computationally simpler approach is proposed that first obtains a symmetric BD-RIS solving a relaxed problem with a closed-form solution, and then projects the relaxed solution onto the set of unitary matrices. The (active) precoders are designed in the second stage according to different criteria that treat the residual interference as noise, such as SVD-based, min-IL, max-SINR, and max-SR precoders. Max-SR precoders are particularly effective, especially when the BD-RIS has a moderate number of elements and the residual interference is still greater than the noise power. When the BD-RIS alone can reduce the IL below the noise floor, the $K$-user MIMO-IC approaches $K$ parallel Gaussian MIMO channels, and the performance of the SVD-based precoders approaches that of the max-SR precoders at a lower computational cost. The study has focused on ICs assisted by a single BD-RIS, which must be optimally located to optimize its area of influence. A future line of study considers multiple BD-RIS, which requires refined signal models accounting for the double/multiple reflections and new optimization algorithms. In addition, it would be interesting to study other criteria for designing the precoders, such as minimizing the transmit power under rate constraints.

\section{Acknowledgment}
This work is supported by the European Commission’s Horizon Europe, Smart Networks and Services Joint Undertaking, research and innovation program under grant agreement 101139282, 6G-SENSES project. The work of I. Santamaria was also partly supported under grant PID2022-137099NB-C43 (MADDIE) funded by MICIU/AEI /10.13039/501100011033 and FEDER, UE.


\renewcommand{\thesubsection}{\Alph{subsection}}
\bibliographystyle{IEEEtran}
\bibliography{main_2column}
\end{document}